\documentclass[a4paper,american,USenglish]{lipics-v2019}
	\def\version{arxiv} 

\usepackage[utf8]{inputenc}
\input{preamble}

\ifarxiv{
	\hideLIPIcs
	\nolinenumbers
}{}

\usepackage{hyperref}
\hypersetup{hidelinks}

\title{Distance Oracles for Interval Graphs via Breadth-First Rank/Select in Succinct Trees}

\titlerunning{Distance Oracles for Interval Graphs via Breadth-First Rank/Select in Trees}

\author{Meng He}{Dalhousie University, Canada}{mhe@cs.dal.ca}{https://orcid.org/0000-0003-0358-7102}{}
\author{J.\ Ian Munro}{University of Waterloo, Canada}{imunro@uwaterloo.ca}{https://orcid.org/0000-0002-7165-7988}{NSERC and the Canada Research Chairs Programme}
\author{Yakov Nekrich}{Michigan Tech, USA}{yakov@mtu.edu}{}{}
\author{Sebastian Wild}{University of Liverpool, UK}{wild@liverpool.ac.uk}{https://orcid.org/0000-0002-6061-9177}{}
\author{Kaiyu Wu}{University of Waterloo, Canada}{k29wu@uwaterloo.ca}{https://orcid.org/0000-0001-7562-1336}{}

\authorrunning{M. He, J.\,I. Munro, Y. Nekrich, S. Wild, K. Wu}
\Copyright{Meng He, J. Ian Munro, Yakov Nekrich, Sebastian Wild, and Kaiyu Wu}



\ifconf{
	\EventEditors{Yixin Cao, Siu-Wing Cheng, and Minming Li}
	\EventNoEds{3}
	\EventLongTitle{31st International Symposium on Algorithms and Computation (ISAAC 2020)}
	\EventShortTitle{ISAAC 2020}
	\EventAcronym{ISAAC}
	\EventYear{2020}
	\EventDate{December 14--18, 2020}
	\EventLocation{Hong Kong, China (Virtual Conference)}
	\EventLogo{}
	\SeriesVolume{181}
	\ArticleNo{57}
	
	\relatedversion{A full version with missing appendices is available~\cite{HeMunroNekrichWildWuArXiv}: \url{https://arxiv.org/abs/2005.07644}.}
}{
	\relatedversion{A shortened version appeared in ISAAC 2020, \doi{10.4230/LIPIcs.ISAAC.2020.57}.}
}

\begin{document}

\maketitle

\begin{abstract}
We present the first succinct distance oracles for (unweighted) interval graphs and 
related classes of graphs, using a novel succinct data structure for ordinal trees
that supports the mapping between preorder (\ie, depth-first) ranks and 
level-order (breadth-first) ranks of nodes in constant time.
Our distance oracles for interval graphs also support navigation queries~-- 
testing adjacency, computing node degrees, neighborhoods, and shortest paths~-- 
all in optimal time. Our technique also yields
optimal distance oracles for proper interval graphs (unit-interval graphs)
and circular-arc graphs.
Our tree data structure supports all operations provided by 
different approaches in previous work, 
as well as mapping to and from level-order ranks and 
retrieving the last (first) internal node before (after) a given node 
in a level-order traversal, all in constant time.
\end{abstract}

%

%
\section{Introduction}
\label{sec:intro}

As a result of the rapid growth of electronic data sets, 
memory requirements become a bottleneck in many applications 
as performance usually drops dramatically as soon as data structures 
do no longer fit into faster levels of the memory hierarchy in computer systems.
Research on \emph{succinct data structures} 
has lead to optimal-space data structures for many types of data~\cite{Navarro2016}.

Graphs are one the most widely used types of data.
In this paper, we study \emph{succinct distance oracles}, \ie,
data structures that efficiently compute the length of a shortest path between
two nodes, for interval graphs and related classes of graphs.
Interval graphs are the intersection graphs of intervals on the real line and 
have applications in operations research~\cite{byfns2000} 
and bioinformatics~\cite{zsfgwkb1994}.
Distance oracles are widely studied; for an overview of the extensive literature 
see~\cite{Sommer2014,z2001,tz2005,pr2014}.

Our distance oracles make fundamental use of (rooted) \emph{trees}.
Standard pointer-based representations of trees use $O(n)$ words or $O(n\log n)$ 
bits to represent a tree on $n$ nodes, but as the culmination of extensive work~\cite{j1989,cm1996,MunroRaman2001,mrr2001,cll2005,mr2004,ly2006,s2007,NavarroSadakane2014,bdmrrr2005,JanssonSadakaneSung2012,bhmr2011,GearyRamanRaman2006,HeMunroRao2012,frr2009}, ordinal trees can be represented \emph{succinctly}, \ie,
using the optimal $2n+o(n)$ bits of space, while supporting a plethora of 
navigational operations in constant time 
(on a word-RAM, which we assume throughout this paper); cf.~\wref{tab:operations}.
One operation that has gained some notoriety for not being supported by any of 
these data structures is 
mapping between \emph{preorder} (\ie, depth-first) ranks and 
\emph{level-order} (breadth-first) ranks of nodes.
Known approaches to represent trees are either fundamentally breadth first~--
like the \textsl{level-order unary degree sequence} ($\LOUDS$)~\cite{j1989}~--
and very limited in terms of supported operation,
or they are depth first~-- like 
the \textsl{depth-first unary degree sequence} (\DFUDS)~\cite{bdmrrr2005},
the \textsl{balanced-parentheses} (\BP) encoding~\cite{MunroRaman2001}
and \textsl{tree covering} (\TC)~\cite{GearyRamanRaman2006}~--
and do not support level-order ranks,
(see \wref{sec:why-level-order-is-hard} for more discussion).

In this paper, we present a new tree data structure that bridges the dichotomy,
solving an open problem of~\cite{HeMunroRao2012}.
Our tree data structure is based on a novel way to (recursively) decompose 
a tree into \emph{forests}
of subtrees that makes computing level-order information possible.
We describe how to support all operations of previous \TC data structures
based on our new decomposition.

Supporting the mapping to and from level-order ranks 
was the missing keystone for our succinct distance oracles for interval graphs,
and our tree data structure will likely be of independent interest
as a building block for future work.

\subparagraph*{Our Results on Trees.} 
Our first result is a succinct representation of ordinal trees which occupies $2n+o(n)$ bits and supports all operations listed in \wref{tab:operations} in $O(1)$ time,
that is, all operations supported by previous work plus these new operations:
\begin{table}[tbh]
	\plaincenter{\fbox{%
	\begin{minipage}{.97\textwidth}\footnotesize
		\newcommand\opitem[2]{\footnotesize#1 &\footnotesize #2 \\[.2ex]}
		\begin{tabular}{p{8.5em}p{33em}}
		\opitem{$\TrParent(v)$}{the parent of $v$, same as $\TrLevAnc(v,1)$}
		\opitem{$\TrDeg(v)$}{the number of children of $v$}
		\opitem{$\TrChild(v,i)$}{the $i$th child of node $v$ ($i\in\{1,\ldots,\TrDeg(v)\}$)}
		\opitem{$\TrChildRank(v)$}{the number of siblings to the left of node $v$ plus $1$}
		\opitem{$\TrDepth(v)$}{the depth of $v$, \ie, the number of edges between the root and $v$}
		\opitem{$\TrLevAnc(v,i)$}{the ancestor of node $v$ at depth $\TrDepth(v)-i$} 
		\opitem{$\TrNbDesc(v)$}{the number of descendants of $v$}
		\opitem{$\TrHeight(v)$}{the height of the subtree rooted at node $v$}
		\opitem{$\LCA(v, u)$}{the lowest common ancestor of nodes $u$ and $v$}
		\opitem{$\TrLeftLeaf(v)$}{the leftmost leaf descendant of $v$}
		\opitem{$\TrRightLeaf(v)$}{the rightmost leaf descendant of $v$}
		\opitem{$\TrLevelLeft(\ell)$}{the leftmost node on level $\ell$}
		\opitem{$\TrLevelRight(\ell)$}{the rightmost node on level $\ell$}
		\opitem{$\TrLevelPred(v)$}{the node immediately to the left of $v$ on the same level}
		\opitem{$\TrLevelSucc(v)$}{the node immediately to the right of $v$ on the same level}
		\opitem{$\TrPrevInt(v)$}{the last internal node before $v$ in a level-order traversal}
		\opitem{$\TrNextInt(v)$}{the first internal node after $v$ in a level-order traversal}
		\opitem{$\TrRank_{X}(v)$}{the position of $v$ in the $X$-order, $X \in \{\pre,\post,\inorder,\DFUDS,\lvl\}$, \ie, in a preorder, postorder, inorder, {\DFUDS} order, or level-order traversal of the tree}
		\opitem{$\TrSelect_{X}(i)$}{the $i$th node in the $X$-order, $X \in \{\pre,\post,\inorder,\DFUDS,\lvl\}$}
		\opitem{$\TrLeafRank(v)$}{the number of leaves before and including $v$ in preorder}
		\opitem{$\TrLeafSel(i)$}{the $i$th leaf in preorder}
	\end{tabular}
	\end{minipage}}}
	\smallskip
	\caption{%
		Navigational operations on succinct ordinal trees. 
		($v$ denotes a node and $i$ an integer).
	}
	\label{tab:operations}
\end{table}

\begin{itemize}
	\item $\TrRank_{\lvl}(v)$ and $\TrSelect_{\lvl}(i)$: 
	computing the position of node $v$ in a level-order traversal of the tree resp.\  finding the $i$th node in the level-order traversal;
	\item $\TrPrevInt(v)$ and $\TrNextInt(v)$: 
	the non-leaf node closest to $v$ in level-order that comes before resp.\ after $v$.
\end{itemize}
Previously, $\TrRank_{\lvl}$ and $\TrSelect_{\lvl}$ were only supported by the $\LOUDS$ representation of trees~\cite{j1989}, which, however, does not support
rank/select by preorder
(and generally only supports a limited set of operations).
Hence our trees are the only succinct data structures to map between preorder (\ie, depth-first) ranks and level-order (breadth-first) ranks in constant time.
\ifarxiv{%
	\wref{tab:suctree} in \wref{app:table} compares our result to previous work.
}{}

\subparagraph*{Our Results on Interval Graphs.}
Interval graphs are intersection graphs of intervals on the line;
several subclasses are obtained by further restricting how the intervals 
can intersect: no interval is properly contained in another ({\em proper interval graphs}), or every interval is contained by (contains) at most $k$ other intervals (\emph{$k$-proper} resp.\ \emph{$k$-improper interval graphs}).
Circular-arc graphs are intersection graphs of arcs on a circle.
The problem of representing these graphs succinctly has been studied by Acan et al.~\cite{AcanChakrabortyJoRao19}, but without efficient distance queries.
We present succinct representations of interval graphs, proper interval graphs, $k$-proper/$k$-improper graphs, and circular-arc graphs in $n\lg n + (5+\epsilon)n + o(n)$, $2n+o(n)$, $2n\lg k + 8n + o(n\log k)$, and $n \lg n + o(n\lg n)$ bits, respectively, 
where $n$ is the number of vertices and $\epsilon>0$ is an arbitrarily small constant,
such that the following operations are supported (time for interval graphs):
\begin{itemize}
\item $\GDegree(v)$: the degree of $v$, \ie, the number of vertices adjacent to $v$;

\item $\GAdjacent(u, v)$: whether vertices $u$ and $v$ are adjacent;

\item $\GNeighbor(v)$: iterating through the vertices adjacent to $v$;

\item $\GSPath(u, v)$: listing a shortest path from vertex $u$ to $v$;

\item $\GDistance(u, v)$: the length of the shortest path from $u$ to $v$;
\end{itemize}
All query times match those of Acan et al.; \GDistance has the same complexity as \GAdjacent; (see \wref{sec:distance-oracles} for precise statements).
Succinctness of our representations (except $k$-(im)proper interval graphs) is evidenced by information-theoretic lower bounds of $n\lg n -2n\lg\lg n - O(n)$ bits~\cite{GavoillePaul2008,AcanChakrabortyJoRao19} and $2n-O(\log n)$ bits~\cite[Thm.\,12]{Hanlon1982} on representing interval graphs (and circular-arc graphs) and proper interval graphs, respectively. 

The best previous distance oracles for interval graphs, proper interval graphs and
circular-arc graphs all result from corresponding
\emph{distance labelings}, a distributed version of distance oracles, 
due to Gavoille et al.~\cite{GavoillePaul2008}. They require
asymptotically $\sim 5 n \lg n$, $\sim 2 n \lg n$, resp.\ $\sim 10 n\lg n$ bits
to represent the labeled graph. 
We improve all of these results even when adding $n\lg n$ bits to store node labels, 
and our data structures further support operations beyond \GDistance.
Interestingly, our distance oracles also prove \emph{separations} 
between distance labelings and distance oracles: 
Our data structures beat corresponding lower bounds for the lengths of distance labelings~--
$3 \lg n-4 \lg \lg n$ for interval graphs~\cite[Thm.\,2]{GavoillePaul2008} resp.
$2 \lg n-2 \lg \lg n-\Oh(1)$ for proper interval graphs \cite[Thm.\,3]{GavoillePaul2008}~--
showing that these ``centralized'' data structures are strictly more powerful
than distributed ones.

\section{Related Work}
\label{sec:related-work}

\subparagraph*{Succinct Representations of Ordinal Trees.}

The $\LOUDS$ representation, first proposed by Jacobson~\cite{j1989} 
and later studied by Clark and Munro~\cite{cm1996} under the word RAM, 
uses $2n+o(n)$ bits to represent a tree on $n$ nodes, such that, 
given a node, its first child, next sibling and parent can be located in constant time. 
Three other approaches, $\BP$, $\DFUDS$ or $\TC$, 
have since been proposed to support more operations while still using $2n+o(n)$ bits. 
As the oldest tree representation after \LOUDS, \BP-based representations have seen a long history of successive improvements and uses in 
various applications of succinct trees. 
The list of supported operations has grown over a sequence of several works~\cite{MunroRaman2001,mrr2001,cll2005,mr2004,ly2006,s2007,NavarroSadakane2014}
to include all standard operations, bar the level-order ones and $\TrRank_{\DFUDS}$ / $\TrSelect_{\DFUDS}$.
The other representations have a similar history, albeit shorter, and we refer to \cite{bdmrrr2005,JanssonSadakaneSung2012,bhmr2011} for $\DFUDS$ and 
\cite{GearyRamanRaman2006,HeMunroRao2012,frr2009} for $\TC$. 
\ifconf{%
	A full survey is also given in the full version of this paper~\cite{HeMunroNekrichWildWuArXiv},
	including a summary of the operations supported by each of these three approaches.%
}{%
	A full survey is also given in \wref{app:table};
	\wref{tab:suctree} there summarizes the operations supported by each of these three approaches.%
}
\ifarxiv{\par}{}
Most works on succinct data structures for trees have focused on \emph{ordinal} trees,
\ie, trees with unbounded degree where the order of children matters, 
but no distinction is made, \eg, between a left and a right single child.
Some ideas have been translated to \emph{cardinal} trees 
(and binary trees as a special case)~\cite{FarzanMunro2014,DavoodiRamanRao2017}.
Other than supporting more operations, work has been done for alternative goals such as achieving compression~\cite{JanssonSadakaneSung2012,FarzanMunro2014}, reducing redundancy~\cite{NavarroSadakane2014} and supporting updates~\cite{NavarroSadakane2014}.

\subparagraph*{Succinct Representations of Graphs.} 
Several succinct representations of (subclasses of) graphs have been studied, \eg, for general graphs~\cite{fm2008}, $k$-page graphs~\cite{j1989}, certain classes of planar graphs~\cite{cghkl1998,cll2005,cds2008}, separable graphs~\cite{bbk2003}, posets~\cite{mn2016} and distributive lattices~\cite{ms2018}. Recently, Acan et al.~\cite{AcanChakrabortyJoRao19} showed how to represent an \emph{interval graph} on $n$ vertices in $n\lg n + (3+\epsilon) n + o(n)$ bits to support $\GDegree$ and $\GAdjacent$ in $O(1)$ time, $\GNeighbor(v)$ in $O(\GDegree(v))$ time and $\GSPath(u, v)$ in $O(|\GSPath(u, v)|)$ time, where $\epsilon$ is a positive constant that can be arbitrarily small. To show the succinctness of their solution, they proved that $n\lg n -2n\lg\lg n - O(n)$ bits are necessary to represent an interval graph.
They also showed how to represent a \emph{proper interval graph} and a \emph{$k$-proper/$k$-improper interval graph} in $2n + o(n)$ and $2n\lg k + 6n + o(n\log k)$ bits, respectively, supporting the same queries.

\subparagraph*{Distance Oracles.} 
Ravi~et al.~\cite{rmr1992} considered the problem of solving the all-pair shortest path problem over interval graphs in optimal $O(n^2)$ time in 1992. 
Later, Gavoille and Paul in 2008 \cite{GavoillePaul2008} designed a labeling scheme on the vertices using $5\lg n + 3$ bit labels to compute the distance between any two vertices $u$, $v$ of an interval graph in $O(1)$ time. 
Their work implies a $5n\lg n + \Oh(n)$ bit distance oracle by simply concatenating all labels. 
Furthermore, they proved a $3\lg n - o(\lg n)$ bit lower bound for distance labeling. 
On the subject of chordal graphs (which contain interval graphs), Singh et al.~\cite{snr2015} designed a data structure of $O(n)$ words that can \emph{approximate} the distance between two vertices $u$ and $v$ in $O(1)$ time, and the answer is between $|\GDistance(u,v)|$ and $2|\GDistance(u,v)|+8$.
More recently, Munro and Wu~\cite{MunroWu18} designed a succinct representation of chordal graphs using $n^2/4+o(n^2)$ bits, which inspired our new distance oracles. 
They also designed an \emph{approximate} distance oracle of $n\lg n + o(n\log n)$ bits with $O(1)$ query time, where answers are within $1$ of the actual distance.

%
%
%
%
%
%
%
%
%
%
%
%
%
%
%
%
%
%
%
%
%
%
%

%
%
%
%
%
%
%
%
%
%
%
%
%
%
%
%
%
%
%
%
%
%
%
%
%
%
%
%
%
%
%
%
%
%
%
%
%
%
%
%
%
%
%
%
%
%
%
%
%
%
%
%
%
%
%

\section{Notation and Preliminaries}
\label{sec:preliminaries}

We write $[n..m] = \{n,\ldots,m\}$ and $[n] = [1..n]$ for integers $n$, $m$.
We use $\lg$ for $\log_2$ and leave the basis of $\log$ undefined (but constant);
(any occurrence of $\log$ outside an Landau-term should thus be considered a mistake).
As is standard in the field, 
all running times assume the word-RAM model with word size $\Theta(\log n)$.

We use the data structure of 
Pǎtraşcu~\cite{Patrascu2008}
for compressed bitvectors:

\begin{lemma}[Compressed bit vector]
\label{lem:compressed-bit-vectors}
	Let $\mathcal{B}[1..n]$ be a
	bit vector of length $n$, containing $m$ $1$-bits.
	For any constant $c$, there
	is a data structure using
	\(
			\lg \binom{n}{m} \wbin+ O\bigl(\frac{n}{\log^c n}\bigr)
		\wwrel\le 
			m \lg \bigl(\frac nm\bigr) \wbin+ O\bigl(\frac{n}{\log^c n}+m\bigr)
	\)
	bits of space that
	supports the following operations in $O(1)$ time
	(for $i \in [1,n]$):
	\begin{itemize}
		\item $\accessop(\mathcal{B}, i)$: return $\mathcal B[i]$, the bit at index $i$ in $\mathcal{B}$.
		\item $\rankop_\alpha(\mathcal{B}, i)$: return the number of bits with
		value $\alpha \in \{0,1\}$ in $\mathcal{B}[1..i]$.
		\item $\selop_\alpha(\mathcal{B}, i)$: return the index of the $i$-th
		bit with value $\alpha \in \{0,1\}$.
	\end{itemize}
\end{lemma}

\section{Tree Slabbing}
\label{sec:slabbing}

In this section, we describe the new tree-covering method used in our data structure.
Throughout this paper, let $T$ be an ordinal tree over $n$ nodes.
We will identify nodes with their ranks $1,\ldots,n$ (order of appearance)
in a preorder traversal.
Tree covering (\TC) relies on a two-tier decomposition: 
the tree consists of mini trees, each of which consists of micro trees.
The former will be denoted by $\mu^i$, the latter by $\mu^i_j$.

\subsection{The Farzan-Munro Algorithm}

We will build upon previously used tree covering schemes.
A greedy bottom-up approach suffices to break a tree of $n$ nodes into $O(n/B)$
subtrees of $O(B)$ nodes each~\cite{GearyRamanRaman2006}. 
However, more carefully designed procedures yield 
restrictions on the touching points of subtrees:
\begin{lemma}[{Tree Covering, \cite[Thm.\,1]{FarzanMunro2014}}]
\label{lem:tree-decomposition}
	For any parameter $B\ge 3$, an ordinal tree with $n$ nodes can be decomposed,
	in linear time, into connected subtrees with the following properties.
	\begin{thmenumerate}{thm:tree-decomposition}
		\item Subtrees are pairwise disjoint except for (potentially) 
			sharing a common subtree root.
		\item Each subtree contain at most $2B$ nodes.
		\item The overall number of subtrees is $\Theta(n/B)$.
		\item Apart from edges leaving the subtree root, at most one other 
			edge leads to a node outside of this subtree.
			This edge is called the ``external edge'' of the subtree.
	\end{thmenumerate}
\end{lemma}
Inspecting  the proof, we can say a bit more:
If $v$ is a node in the (entire) tree and is also the root of several subtrees 
(in the decomposition), then the way that $v$'s children (in the entire tree) are
divided among the subtrees is into \emph{consecutive} blocks.
Each subtree contains at most two of these blocks. 
(This case arises when the subtree root has 
exactly one heavy child: a node whose subtree size is greater than $B$, in the decomposition algorithm.)

\subparagraph*{Why is level-order rank/select hard?}
\label{sec:why-level-order-is-hard}

Suppose we try to compute the level-order rank of a node $v$, and we try to reduce the
global query (on the entire tree $T$) to a local query that is constrained to a mini tree $\mu^i$.
This task is easy if we can afford to store the level-order ranks of the leftmost node in $\mu^i$ 
\emph{for each level} of $\mu^i$: then the level-order rank of $v$ is simply 
the global level-order rank of $w$, where $w$ is the leftmost node in $\mu^i$ on $v$'s level
($v$'s depth), 
plus the local level-order rank of $v$,
minus the local level-order rank of $w$ minus one 
(since we double counted the nodes in $\mu^i$ on the levels above $w$).

However, for general trees, we cannot afford to store the level-order rank of all leftmost nodes.
This would require $\TrHeight(\mu^i)\cdot \lg n$ bits for $\TrHeight(\mu^i)$ the height of $\mu^i$; towards a sublinear overhead in total,
we would need a $o(1)$ overhead per node, 
which would (on average) require $\mu^i$ to have $|\mu^i| = \omega(\TrHeight(\mu^i) \log n)$ nodes
or height $\TrHeight(\mu^i) = o(|\mu^i| / \log n)$.
Since the tree $T$ to be stored can be one long path 
(or a collection of few paths with small off-path subtrees etc.), 
any approach based on decomposing $T$ into induced subtrees is bound to fail the 
above requirement.

The solution to this dilemma is the observation that the above ``bad trees''
have another feature that we can exploit: The total number of nodes 
on a certain interval of levels is small.
If we keep such an entire horizontal slab of $T$ together,
translating global level-order rank queries into local ones does not need
the ranks of all leftmost nodes: everything in these levels is entirely contained
in $\mu^i$ now, and it suffices to add the level-order rank of the (leftmost) root in $\mu^i$.

Our scheme is based on decomposing the tree into parts that are 
one of these two extreme cases~-- ``skinny slabs'' or ``fat subtrees''~--
and counting them separately to amortize the cost for storing 
level-order information.

\subsection{Covering by Slabs}

We fix two parameters: $H \in\N$, the \underline height of slabs,
and $B > H$, the target \underline block size.
We start by cutting $T$ horizontally into slabs of thickness/height exactly $H$,
but we allow ourselves to start cutting at an offset $o\in[H]$.
We choose $o$ so as to minimize the total number of nodes on levels 
at which we make the horizontal cuts.
We call these nodes \emph{s-nodes} (``\underline slabbed nodes''), 
and their parent edges \emph{slabbed edges}.
A simple counting argument shows that the number of $s$-nodes (and slabbed edges)
is at most $n / H$.

We will identify induced subgraphs with the set of nodes that they are induced by.
So 
$S_i = \bigl\{v : \TrDepth(v) \in [{(i-1)H +o}\mathrel{\;..\;}{iH +o}]\bigr\}$,
the set of nodes making up the $i$th slab,
also denotes the $i$th slab itself, $i=0,\ldots,h$.
Obviously, the number of slabs is $h+1 \le n / H + 2$.
We note that the $s$-nodes are contained in \emph{two} slabs. 
For any given slab, we will refer to the first $s$-level included as 
(original) $s$-nodes and the second as \emph{promoted} $s$-nodes. 
Note that the first slab does not contain any $s$-nodes and the last slab 
does not contain promoted $s$-nodes.

Since $S_i$ is (in general) a \emph{set} of subtrees, 
ordered by the left-to-right order of their roots,
we will add a \emph{dummy root} to turn it into a single tree.
We note that the $s$-nodes are the first (after the dummy root) 
and the last levels of any slab.

If $|S_i| \le B$, $S_i$ is a \emph{skinny} subtree (after adding the dummy root)
and will not be further subdivided.
If $|S_i| > B$, we apply the Farzan-Munro tree-covering scheme 
(\wref{lem:tree-decomposition}) with parameter $B$ 
to the slab (with the dummy root added) to obtain \emph{fat subtrees}.
This directly yields the following result; an example is shown in \wref{fig:tree-slabbing-example}.

\begin{figure}[thb]
	\resizebox{\linewidth}!{\input{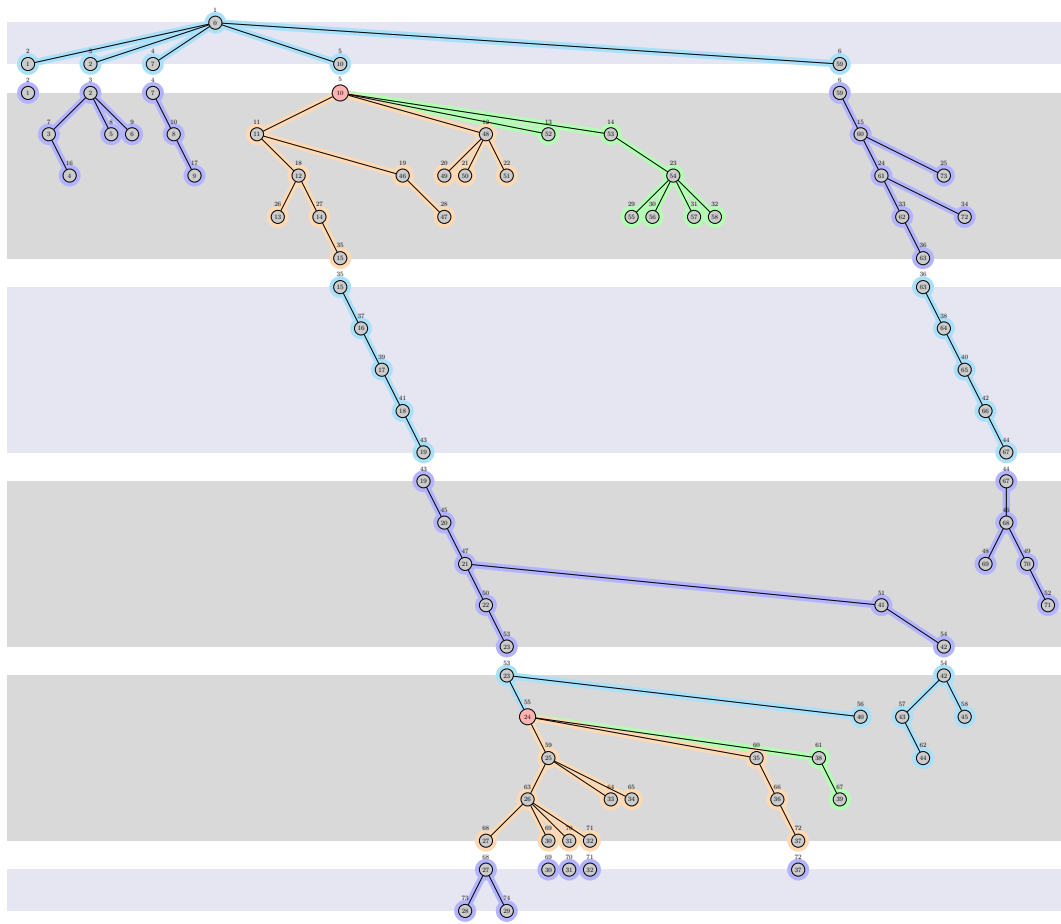}}
	\caption{%
		An example of the tree-slabbing decomposition from \wref{thm:tree-slabbing} 
		with $B=11$ and $H=4$.
		Slabs are shown as shaded areas (light blue for skinny slabs, light gray for fat slabs).
		All s-nodes are depicted twice, one in each slab they belong to.
		The trees within a slab are connected by a dummy root (not depicted) and further 
		decomposed as in \wref{lem:tree-decomposition};
		the resulting subtrees are shown by the edge colors.%
	}
	\label{fig:tree-slabbing-example}
\end{figure}

\begin{theorem}[Tree Slabbing]
\label{thm:tree-slabbing}
	For any parameters $B > H \ge 3$, an ordinal tree $T$ with $n$ nodes 
	can be decomposed, in linear time, 
	into connected subtrees with the following properties.
	\begin{thmenumerate}{thm:tree-slabbing}
		\item Subtrees are pairwise disjoint except for (potentially) 
			sharing a common subtree root.
		\item Subtrees have size $\le M = 2B$ and height $\le H$.
		\item Every subtree is either \emph{pure} 
			(a connected induced subgraph of $T$),
			or \emph{glued} 
			(a dummy root, whose children are connected induced subgraphs of $T$).
		\item Every subtree is either a \emph{skinny (slab)} subtree
			(an entire slab) or \emph{fat}.
		\item The overall number of subtrees is $\Oh(n/H)$,
			among which $\Oh(n/B)$ are fat.
		\item Connections between subtrees $\mu$ and $\mu'$ 
			are of the following types:
		\begin{enumerate}[
			label={\textsf{\textbf{\color{darkgray}{\makebox[\widthof{(ii)}][c]{%
			\arabic*.}}}}},
			leftmargin=1.5em,
			topsep=0.5ex,
		]
			\item $\mu$ and $\mu'$ share a common root. Each subtree contains at most
				two blocks of consecutive children of a shared root.
			\item The root of $\mu'$ is a child of the root of $\mu$.
			\item The root of $\mu'$ is a child of another node in $\mu$.
				This happens at most once in $\mu$.
			\item $\mu'$ contains the original copy of a promoted s-node in $\mu$.
				The total number of these connections is $\Oh(n/H)$.
		\end{enumerate}
	\end{thmenumerate}
\end{theorem}

\subparagraph*{,,Oans, zwoa, G'suffa.``}

The above tree-slabbing scheme has two parameters, $H$ and $B$.
We will invoke it \emph{twice}, 
first using $H  = \lceil\lg^3 n \rceil$ and $B = \lceil \lg^5 n\rceil$ to form
$m$ mini trees $\mu^1,\ldots,\mu^m$ of at most $M = 2B$ nodes each.
While in general we only know $m = \Oh(n/H) = \Oh(n/\log^3 n)$,
only $\Oh(n/M) = \Oh(n/\log^5 n)$ of these mini trees are 
\emph{fat} subtrees (subtrees of a fat slab),
the others being skinny.
Mini trees $\mu^i$ are recursively decomposed by
tree slabbing with height $H'  = \lceil \frac{\lg n}{(\lg\lg n)^2} \rceil$
and block size $B' = \lceil \frac18 \lg n\rceil$ 
into micro trees $\mu^i_1,\ldots,\mu^i_{m'_i}$ of size at most
$M' = 2b = \frac14 \lg n$.
The total number of micro trees is $m' = m'_1+\cdots+m'_{m} = \Oh(n/H')$,
but at most $\Oh(n/B')$ are fat micro trees.
We refer to the $s$-nodes created at mini resp.\ micro tree level
as \emph{tier-1} resp.\ \emph{tier-2} $s$-nodes. 
After these two levels of recursion we have reached a size for 
micro trees small enough to use a ``Four-Russian''
lookup table (including support for various micro-tree-local operations)
that takes sublinear space.

\subparagraph*{Internal node ids.}

Internally to our data structure, we will identify a node $v$ by 
its ``$\tau$-name'',
a triple specifying the mini tree, the micro tree within the mini tree, 
and the node within the micro tree.
More specifically,
$\tau(v) = \langle \tau_1,\tau_2,\tau_3\rangle$ means that $v$
is the $\tau_3$th node in the micro-tree-local preorder (DFS order) traversal of $\mu^{\tau_1}_{\tau_2}$;
mini trees are ordered by when their first node
appears in a preorder traversal of $T$, ties (among subtrees sharing roots) broken
by the second node, and similarly for micro trees inside one mini tree.

Since there are $O(n/H)$ mini trees, $O(B/H')$ micro trees 
inside one mini tree, and $O(B')$ nodes in one micro tree,
we can encode any $\tau$-name with 
$\sim \lg n + 2 \lg \lg n  +2\lg\lg\lg n$ bits.
The concatenation $\tau_1(v)\tau_2(v)\tau_3(v)$ can be seen as a binary number;
listing nodes in increasing order of that number gives the \emph{$\tau$-order}
of nodes.

\subparagraph*{Who gets promotion?}

A challenge in tree covering is to handle operations like \TrChild 
when they cross subtree boundaries.
The solution is to add the endpoint of a crossing edge also
to the parent mini/micro tree; 
these copies of nodes are called \emph{(tier-1/tier-2) promoted nodes}.
They have their own $\tau$-name, but actually refer to the same original node;
we call the $\tau$-name of the original node the \emph{canonical $\tau$-name}.

For tree slabbing, we additionally have slabbed edges to handle.
As mentioned earlier,
we promote \emph{all} endpoints of slabbed edges into the parent slab
\emph{before we further decompose a slab}.
That way, the size bounds for subtrees already include any promoted copies,
but we blow up the number of subtrees 
by an~-- asymptotically negligible~-- factor of $1+1/H\sim 1$.
Promoted s-nodes again have both canonical and secondary $\tau$-names.

\section{Operations on Slabbed Trees}
\label{sec:operations}

We now describe how to support operations efficiently in our data structure.
\ifconf{%
	Due to space constraints, we describe some exemplary ones here
	and defer the others to the full version of this paper~\cite{HeMunroNekrichWildWuArXiv}.%
}{%
	We describe some exemplary ones here
	and defer the others to \wref{app:more-operations}.%
}

We start by describing some common concepts.
The \emph{type} of a micro tree is the concatenation of its size (in Elias code), 
the \BP of its local shape, 
and the preorder rank of the promoted dummy node (0 if there is none),
and several bits indicating whether the lowest level are promoted $s$-nodes, 
and whether the root is a dummy root.
We store a variable-cell array of the \emph{types} of all micro trees
in $\tau$-order.
The \BP of all micro trees will sum to $2n+\Oh(n/H') = 2n + o(n)$ bits of space;
the other components of the type are asymptotically negligible.
A type consists of at most $\sim \frac12 \lg n$ bits, 
so we can store a table of all possible
types with various additional precomputed local operations in 
$\Oh(\sqrt n \operatorname{polylog}(n))$~bits.

\subsection{Preorder rank/select}

We first consider how to convert between global preorder ranks and $\tau$-names.
Let us fix one level of subtrees, say mini trees. 
Consider the sequence $\tau_1(v)$ for all the nodes $v$ in a preorder traversal.
A node $v$ so that $\tau_1(v) \ne \tau_1(v-1)$
is called a \emph{(tier-1) preorder changer} \cite[Def.\,4.1]{HeMunroRao2012}.
Similarly, nodes $v$ with $\tau_2(v) \ne \tau_2(v-1)$ are called
\emph{(tier-2) preorder changers}.
We will associate with each node $v$ ``its'' tier-1 (tier-2) preorder changer $u$,
which is the last preorder changer preceding $v$ in preorder, \ie, 
$\max \{u \in [1..v] : \tau_1(u) \ne \tau_1(u-1) \}$;
(Recall that we identify nodes with their preorder rank.)

By Theorem~\ref{thm:tree-slabbing}, the number of tier-1 preorder changers is
$O(n/H)$, since the only times a mini-tree can be broken up is through the external edge 
(once per tree), the two different blocks of children of the root, or at slabbed edges.
Similarly, we have $O(n/H')$ tier-2 preorder changers.
We can thus store a compressed bitvector (\wref{lem:compressed-bit-vectors})
to indicate which nodes in a preorder traversal are (tier-1/tier-2) preorder changers.
The space for that is 
$O(\frac nH \log (H) + n \frac{\log \log n}{\log n}) = o(n)$ for tier 1
and  
$
		O(\frac n{H'}\log H' + n \frac{\log \log n}{\log n}) 
	=	O(n \frac{(\log \log n)^3}{\log n}) = o(n)
$ for tier 2.

We will additionally store a compressed bitvector indicating preorder changers
by $\tau$-name, \ie, we traverse all nodes in $\tau$-order and add a $1$ if the current
node is a preorder changer, and a $0$ if not.
We can afford to do this using \wref{lem:compressed-bit-vectors} for 
tier-1 and tier-2 in $o(n)$ bits.
(The universe grows to $n \operatorname{polylog}(n)$, 
but with sufficiently large~$c$ that does not affect 
the space by more than a constant factor).
We can store $O(\log n)$ bits for each tier-1 changer
and $O(\log\log n)$ bits for each tier-2 changer in an array, 
and using rank on the above bitvectors,
we can access that information given the node's global preorder or $\tau$-names.

\subparagraph{Select} 

Given the preorder number of a node $v$, we want to find $\tau(v)$.
Let $u$ and $u'$ be the tier-1 resp.\ tier-2 preorder changers associated
with $v$.
The core observation is that $\tau_1(u) = \tau_1(v)$ and $\tau_2(u') = \tau_2(v)$,
since a node's tier-1 (tier-2) preorder changer by definition lies in the same
mini- (micro-) tree as $v$.
We thus store the array of $\tau_1$-numbers of all tier-1 preorder changers
as they are visited by a preorder traversal of $T$.
Using rank and select on the bitvectors from above, 
we find $u$, for which we look up $\tau_1$.
The procedure applies, \textit{mutatis mutandis}, to $\tau_2$ using the 
tier-2 preorder changer $u'$.
Since $\tau_2$ is local to a mini tree, $\lg M = O(\log \log n)$ bits suffice,
so we can afford to store $\tau_2$ for every tier-2 changer.
We also store the $\tau_3$-number for each tier-2 changer in the same space. 
We can then obtain $\tau_3(v)$ as the sum of $\tau_3(u')$
and the distance from the last 1 in the bit vector indicating tier-2 changers.

\subparagraph{Rank} 
Given $\tau(v)=\langle\tau_1,\tau_2,\tau_3\rangle$, find the global preorder rank.
Let again $u$ and $u'$ be the tier-1 resp.\ tier-2 preorder changers associated
with $v$. The idea is to compute the preorder rank as $u + (u'-u) + (v-u')$,
\ie, the global preorder of $u$ and the distances between $u$ and $u'$ resp.\ 
$u'$ and $v$.
Of course, we do not know $u$ and $u'$ or their distances directly, 
but we can store them as follows.
We use the $\tau$-order of nodes
to store the mapping from $\tau$-name of tier-1 preorder changers to their
global preorder ranks.
For each tier-2 changer, we store the mapping of $\tau$-names to distances 
to associated tier-1 changers ($O(\log\log n)$ bits each).

It remains to compute $\tau(u)$ and $\tau(u')$ from $\tau(v)$.
$v$ and $u'$ only differ in $\tau_3$ and we use the micro-tree lookup table
to store $\tau_3$ of each node's tier-2 changer.
Then, we store for each tier-2 changer $u'$ the pair $\langle\tau_2,\tau_3\rangle$
of its tier-1 changer  (another $O(\log\log n)$ bits each).
Using the $\tau$-names of $u$ and $u'$, we obtain the preorder rank of $v$.

\subsection{Level-order rank/select}

Let $w_1,\ldots,w_n$ be the nodes of $T$ in level order,
\ie, $w_i$ is the $i$th node visited in the left-to-right breadth-first traversal
of $T$.
Similar to the preorder, we call a node $w_i$ a \emph{tier-1 (tier-2) level-order changer}
if $w_{i-1}$ and $w_i$ are in different mini- (micro-) trees.
The following lemma bounds the number of tier-1 (tier-2) level-order changers.

\begin{lemma}
	\label{lem:levelorder-changer}
The number of tier-1 (tier-2) level-order changers is 
$\Oh(n/H + n H/B) = \Oh(n/\log^2 n)$ ($\Oh(n/H'+nH'/B') = \Oh(n/(\log\log n)^2)$).
\end{lemma}
\begin{proof}%
	We focus on tier 1; tier 2 is similar.
	Lemma~\ref{thm:tree-slabbing} already contains all ingredients: A skinny-slab subtree
	consists of an entire slab, so its nodes appear contiguous in level order.
	Each skinny mini tree thus contributes only 1 level-order changer, for a total of $\Oh(n/H)$
	For the fat subtrees, each level appears contiguously in level order, and within a level,
	the nodes from one mini tree form at most 3 intervals: one gap can result from a child
	of the root that is in another subtree, splitting the list of root children into two intervals,
	and a second gap can result from the single external edge.
	The other connections to other mini trees are through s-nodes, 
	and hence all lie on the same level.
	So each fat mini tree contributes at most 3 changers per level it spans, for a total of
	$\Oh(H\cdot n/B)$ level-order changers.
\end{proof}
With that preparation done, we proceed similarly as for preorder.

\subparagraph{Select}
Given the level-order rank $i$, find $\tau(w_i)$.
We store $\tau_1(w_1),\ldots,\tau_1(w_n)$ in a piece-wise constant array,
using the same technique as for preorder
(compressed bitvector for changers, explicit values at changers),
and similarly for $\tau_2(w_1),\ldots,\tau_2(w_n)$.
Both require $o(n)$ bits.

For $\tau_3$, we have to take an extra step as we don't visit nodes in preorder now.
But we can store the micro-tree-local \emph{level-order} rank $j'$ at all
tier-2 level-order changers and compute the distance $j''$ of $w_i$ from its
tier-2 changer.
The sum $j'+j''$ is the micro-tree-local level-order rank of $w_i$, which we
translate to $\tau_3(w_i)$ using the lookup table.

\subparagraph{Rank}
Given a node $v$ by $\tau$-name, we now seek the $i$ with $v=w_i$.
We compute $i$
as $j + (j'-j) + (i-j')$ for $w_j$ and $w_{j'}$ the tier-1 resp.\ tier-2 
level-order changers of $v=w_i$; 
(this is similar as for preorder rank above).

From the micro-tree lookup table, we obtain $\tau_3(w_{j'})$
and the level-order distance to~$v$.
For tier-2 changers, we store the mapping from $\tau$ to distance
(in level order) to their tier-1 changers, 
as well as $\langle \tau_2,\tau_3\rangle$ of their tier-1 changers.
Finally, for tier-1 changers, we map $\tau$ to their lever-order ranks.
That determines all summands for $i$.

\subsection{Previous Internal Node in Level Order}

Given $\tau(v)$, find $\TrPrevInt(v) = \tau(w)$, 
where $w$ is the last non-leaf node ($\TrDeg(w)>0$) preceding $v$ in level order.
In the micro-tree lookup table, we store whether 
there is an internal node to the left of $v$ inside the micro-tree, 
and if so, its $\tau_3$.
If $w$ does not lie in $\mu^{\tau_1(v)}_{\tau_2(v)}$,
we get $v$'s tier-2 level-order changer $u'$ from the lookup table,
for which we store whether there is an internal node to the left of $u'$ inside the micro-tree, 
and if so, store its $\langle\tau_2,\tau_3\rangle$.
If $w$ is also not in $\mu^{\tau_1(v)}$,
we move to $u'$'s tier-1 level-order changer 
($\langle\tau_2(u),\tau_3(u)\rangle$ is stored for $u'$).
At tier-1 changers $u$, we store $\TrPrevInt(u)$ directly.

\medskip
\noindent
Combining our work in Sections~\ref{sec:slabbing}, \ref{sec:operations}, 
\ifconf{%
	and the appendix of the full version of this paper~\cite{HeMunroNekrichWildWuArXiv}%
}{%
	and \wref{app:more-operations},%
}
we have our first result:
\begin{theorem}[Succinct trees]
\label{thm:levelorder-trees}
An ordinal tree on $n$ nodes can be represented in $2n + o(n)$ bits to support all the tree operations listed in \wref{tab:operations} in $O(1)$ time. 
\end{theorem}

\section{Distance Oracles and Interval Graph Representations}

\label{sec:distance-oracles}

In this section, we present new time- and space-efficient distance oracles
for interval graphs and related classes.
Here (and throughout this paper), we assume an interval realization of the graph
$G=([n],E)$ is given where all endpoints are disjoint and lie in $[2n]$;
such can be computed efficiently from $G$~\cite{AcanChakrabortyJoRao19}.
Vertices of an interval graph are labeled $1,\ldots,n$, 
sorted by the left endpoints of their intervals.

\subsection{Distances in Interval Graphs}

We first describe how to augment an interval-graph representation  
with $\Oh(n)$ additional bits of space to support $\GDistance$ in constant time.
Our distance oracles are based on the graph data structures of 
Acan et al.~\cite{AcanChakrabortyJoRao19};
we recall their result for interval graphs.
\begin{lemma}[Succinct interval graphs, {\cite{AcanChakrabortyJoRao19}}]
\label{lem:succinct-interval}
	An interval graph can be represented using $n\lg n + (3+\epsilon)n + o(n)$ bits to support {\GAdjacent} and \GDegree in $O(1)$ time, \GNeighbor in $O(\GDegree(v))$ time and $\GSPath(u,v)$ in $O(\GDistance(u,v))$ time. Moreover, the interval $I_v = [\ell_v,r_v]\in[2n]^2$ representing a vertex can be retrieved in $O(1)$ time.%
\footnote{
	Note that the arXiv version~\cite{AcanChakrabortyJoRao19arxiv} of 
	\cite{AcanChakrabortyJoRao19} erroneously claims a space usage of 
	$n\lg n + (2+\epsilon)n + o(n)$ bits for their data structure. 
	Interestingly, it is indeed
	possible to reduce the space
	to that by storing $r_1,\ldots,r_n \in [2n]$, 
	the right endpoints, in rank-reduced form, $R[1..n]$,
	(a permutation)
	and using $r_i = \selop_1(S,R[i])$.%
}
\end{lemma}
As interval graphs are a subclass of chordal graphs, we will be using the algorithm of Munro and Wu~\cite{MunroWu18} to compute distances. For a vertex $v$, denote the \emph{bag} of $v$ by $B_v = \{ w : \ell_v \in I_w \}$, \ie, the set of vertices whose intervals contain the left endpoint of $v$'s interval. As in~\cite{MunroWu18}, we define $s_v = \min B_v$. 
The shortest path algorithm given in~\cite{MunroWu18} is similar to the one in~\cite{AcanChakrabortyJoRao19}. 
Given $u < v$, we compute the shortest path by
checking if $u$ and $v$ are adjacent. If so, add $u$ to the path;
otherwise, add $s_v$ to the path and recursively find $\GSPath(s_v, u)$.

As the next step for every vertex $v$ is the same regardless of destination $u$, 
we can store this unique step for each vertex as the parent pointer of a tree.
We construct a tree $T$ as follows: 
for every vertex $v=1,\ldots,n$ (in that order), 
add node $v$ to the tree as the rightmost (last) child of $s_v$;
see \wref{fig:interval_graph} for an example. 
The node $v=1$ is the root of the tree. 
Thus we have identified each vertex of $G$ with a node of $T$.
This correspondence is captured by \wref{lem:correspondence} below.

\begin{figure}[htb]
	\scalebox{.85}{%
	\begin{tikzpicture}[xscale=0.6, yscale=.5,inner sep=2.5pt]
		\node at (1,0) {5};
		\node at (1,1) {4};
		\node at (1,2) {3};
		\node at (1,3) {2};
		\node at (1,4) {1};
		\node at (1,-1) {6};
		\foreach \x in {1,...,12} {
			\draw[thin,dotted] (\x+1,4) -- ++(0,-6) node[below] {\tiny \x} ;
		}
		
		\begin{scope}[very thick,|-|]
		\draw[black] (2,4) -- (7,4);	
		\draw[black] (3,3) -- (5,3);
		\draw[black] (4,2) -- (9,2);		
		\draw[black] (6,1) -- (11,1);		
		\draw[black] (8,0) -- (12,0);		
		\draw[black] (10,-1) -- (13,-1); 
		\end{scope}
		
		\node[draw,circle,black] (1) at (16,1) {1};
		\node[draw,circle,black] (2) at (16,3) {2};
		\node[draw,circle,black] (3) at (18,1) {3};
		\node[draw,circle,black] (4) at (18,3) {4};
		\node[draw,circle,black] (5) at (20,3) {5};
		\node[draw,circle,black] (6) at (20,1) {6};
		\draw[black] (1) -- (2) -- (3) -- (4) -- (1) -- (3);
		\draw[black] (2) --(4);	
		\draw[black] (3) -- (5) -- (4);
		\draw[black] (5) -- (6) -- (4); 
		
		\node[draw,circle,black] (7) at (25,4) {1};
		\node[draw,circle,black] (8) at (23,2) {2};
		\node[draw,circle,black] (9) at (25,2) {3};
		\node[draw,circle,black] (10) at (27,2) {4};
		\node[draw,circle,black] (11) at (25,0) {5};
		\node[draw,circle,black] (12) at (27,0) {6};
		\draw[black] (7) -- (8);
		\draw[black] (7) -- (9);
		\draw[black] (7) -- (10);
		\draw[black] (10) -- (12);
		\draw[black] (9) -- (11);
	\end{tikzpicture}%
	}
	\caption{An Interval Graph (middle) with Interval Representation (left), and distance tree constructed (right).}
	\label{fig:interval_graph}
\end{figure}
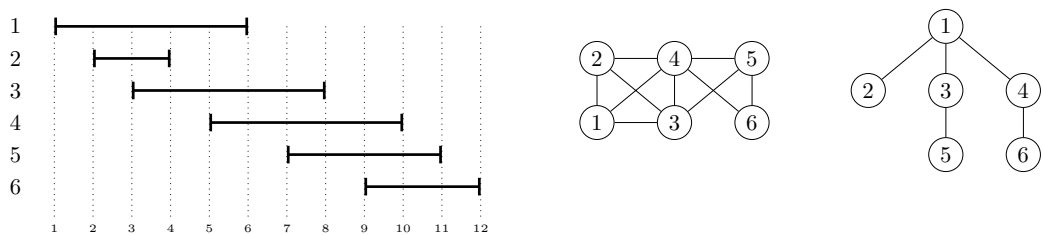

We note that the above construction is undefined for a disconnected graph, as the leftmost interval of a component would have an undefined parent. The simplest way to solve this is to set the parent of such a vertex $v$ as $v-1$ (that is we add the edge between them). We will also need to include a length $n$ bit-vector, where the $i$th entry is a 1 if vertex $i$ is the first vertex of a component (to keep track of the edges we added). Any distance queries (between $u$ and $v$) will first check if the two vertices are in the same component by performing a rank query on the bit-vector at indices $u$ and $v$, and check that they are the same. Similarly for adjacency and neighborhood queries; we will need to check if vertices are the first vertex of a component, and if so, make sure the added edge is not reported.

\begin{lemma}[Distance tree BFS]
	\label{lem:correspondence}
	Let $a_1,a_2,\ldots, a_n$ be a \emph{breadth-first} traversal of $T$. 
	Then the corresponding vertices of $G$ are $1,2\ldots n$.
\end{lemma}
\begin{proof}%
	First note that it immediately follows from the incremental construction of $T$ in level order that
	the node with largest index inserted so far is always the rightmost node on the deepest level of $T$.
	So if the graph is disconnected, our procedure above does not change the order of the vertices in level order, nor the order of the vertices in $G$. 
	So we may assume that the graph is connected.
	
	For vertices $u < v$, we will show that the node in $T$ corresponding to $u$ appears before the corresponding node to $v$ in $T$ in level order.
	
	Suppose by contradiction that it is not. Thus we must have that $s_v < s_u$ in order for it to be before $u$ in the breadth-first ordering. If $s_v = s_u$, then they are siblings and $v$ is added to the right of $u$ by construction.
	
	Therefore, we have the following facts: i) $\ell_v > \ell_u$ as $v > u$, ii) $\ell_v \in I_{s_v}$ by definition of $s_v$, iii) $\ell_u \in I_{s_u}$ by definition of $s_u$, and iv) $\ell_{s_v} < \ell_{s_u}$ as $s_v < s_u$.
	Thus we have $\ell_{s_v} < \ell_{s_u} < \ell_u < \ell_v < r_{s_v}$, and thus $\ell_u \in I_{s_v}$. By definition, $s_v \in B_u$ which contradicts the fact that $s_u = \min B_u$.
\end{proof}
With this correspondence, we will abuse notation when the context is clear and refer to both the vertex in the graph and the corresponding node in the tree by $v$. Any conversion that needs to be done will be done implicitly using $\TrRank_{\lvl}$ and $\TrSelect_{\lvl}$.
Now consider the shortest path computation for $u < v$. 
The only candidates potentially adjacent to $u$ are the ancestors of $v$ at depths $\TrDepth(u)-1$, $\TrDepth(u)$, and $\TrDepth(u)+1$. The ancestor $z$ of $v$ at depth $\TrDepth(u) + 2$ cannot be adjacent to $u$ as $w=\TrParent(z) > u$, and $\TrParent(z)$ is defined as the smallest node adjacent to $z$.
Thus the distance algorithm reduces to the following: For vertices $u < v$, compute $w=\TrLevAnc(v, \TrDepth(u)+1)$, the ancestor of $v$ at depth $\TrDepth(u)+1$. Find the distance between $u$ and $w$ using the $\GSPath$ algorithm. This is at most 3 steps, so in $O(1)$ time. Finally take the sum of the distances, one from the difference in depth and the other from the $\GSPath$ algorithm.
The extra space needed is to store the tree $T$, using $2n + o(n)$ bits,
and for disconnected graphs, the component bitvector.

The results described above are summarized in the following theorem:
\begin{theorem}[Succinct interval graphs with distance]
\label{thm:succinct-interval-graphs}
	An interval graph $G$ can be represented using $n\lg n + (5+\epsilon)n + o(n)$ bits to support 
	\GAdjacent, \GDegree and \GDistance in $O(1)$ time, 
	\GNeighbor in $O(\GDegree(v)+1)$ time, and 
	$\GSPath(u,v)$ in $O(\GDistance(u,v)+1)$ time.
	If $G$ is disconnected, the space needed is $n\lg n + (6+\epsilon)n + o(n)$ bits.
\end{theorem}
Finally we note that this augmentation can without changes 
be applied to subclasses of interval graphs; 
we thus obtain the following theorem:
\begin{theorem}[Succinct $k$-proper/-improper interval graphs with distance]
\label{thm:succinct-k-proper-interval-graphs}
	~\\A $k$-proper ($k$-improper) interval graph%
	\footnote{%
		We note that Klavík et al.~\cite{KlavikOtachiSejnoha2019} consider a closely 
		related class of interval graphs, $k$-\textsf{NestedINT} that is similar to 
		(and contains) Acan et al.'s~\cite{AcanChakrabortyJoRao19} class of 
		$(k-1)$-improper interval graphs, 
		but defines $k$ as the length of longest chain of pairwise nested intervals.
		The data structures of Acan et al.\ directly apply to this notion
		by adapting the definition of $S'$.
	}
	$G$ can be represented using $2n\lg k + 8n + o(n\log k)$ bits to support 
	$\GDegree$, $\GAdjacent$, $\GDistance$ in $O(\log\log k)$ time, 
	$\GNeighbor$ in $O(\log\log k \cdot(\GDegree(v)+1))$ time and 
	$\GSPath(u,v)$ in $O(\log\log k  \cdot(\GDistance(u,v)+1))$ time.
	If $G$ is disconnected, the space needed is $2n\lg k + 9n + o(n\log k)$ bits.
\end{theorem}
The additional space is a lower-order term if $k = \omega(1)$.
While Acan et al.'s data structure is not succinct, either, for $k=\Oh(1)$,
a different tailored representation for proper interval graphs ($k=0$)
is presented there.
Here, simply adding our distance tree is not good enough.

\subsection{Succinct Proper Interval Graphs with Distance}
Recall that a proper interval graph is an interval graph that admits an interval representation with no interval
properly contained in another.
As before, each vertex $v$ is associated with an interval $I_v$ and vertices sorted by left endpoints. 
The information-theoretic lower bound for this class of graphs is $2n - O(\log n)$ bits~\cite[Thm.\,12]{Hanlon1982}.
Hanlon also shows that asymptotically, a $0.626578$-fraction of all proper interval graphs
is connected, so the same lower bound holds for connected proper interval graphs.

While adding the distance tree on top of the existing representation is too costly,
our the key insight here is that the graph can be \emph{recovered} from 
the distance tree, and indeed, 
we can answer all graph queries directly on the latter.
Thus for connected proper interval graphs, the representation is succinct, but an extra $n+o(n)$ bits is required for disconnected proper interval graphs in the worst case. However, if the number of components is not too large, say $O(n/\log(n))$ components, our redundancy remains $o(n)$ using \wref{lem:compressed-bit-vectors}. We will assume that the graph is connected, and use the extra steps required as described in the general interval graph case.
First, the neighborhood of a vertex can be succinctly described:
\begin{lemma}
	\label{lem:neighbourhood}
	Let $v$ be a vertex in a proper interval graph. Then there exists vertices $u_1 \le u_2$ such that the (closed) neighborhood of $v$ is equal to 
	the vertices in $[u_1, u_2]$.
\end{lemma}
\begin{proof}%
	Let $u_1 < v$ be adjacent to $v$. Let $w = u_1 +1$. As $G$ is a proper interval graph, we have the following inequalities: $ \ell_{u_1} < \ell_w \le \ell_v < r_{u_1} < r_w .$
	Thus $I_v$ intersects $I_w$ and $v$ is adjacent to $w$. So the neighborhood of $v$ consisting of vertices with smaller label forms a contiguous interval.
	
	Similarly, the same argument can be made for the vertices with larger labels.
\end{proof}
Let $T$ be the tree constructed in the previous section. We already showed how to compute $\GSPath$ and $\GDistance$ for $G$ (based on an implementation of $\GAdjacent$). We now show how to compute $\GAdjacent$, $\GDegree$ and $\GNeighbor$.

$\GAdjacent$: Let $u < v$. We first check if $v$ is the leftmost node in its component; if so, $u$ and $v$ cannot be adjacent.
Otherwise, we compute $s_v$ (using \TrParent); then $u$ and $v$ are adjacent iff $s_v \le u$.
Correctness follows from the fact that the neighborhood of $v$ is a contiguous interval.

$\GNeighbor$: Let the neighborhood of $v$ be $[u_1,u_2]$. By the definition of $s_v$, we have that $u_1 = s_v$ (unless $v$ is leftmost; then $u_1=v$). Thus it remains to compute $u_2$.
If $v$ is rightmost in its component, $u_2=v$; otherwise we find $u_2$
using the following lemma in $O(1)$ time.

\begin{lemma}
	\label{lem:proper-neighbourhood}
	If $v$ is a leaf, then $u_2 = \TrLastChild(\TrPrevInt(v))$;
	otherwise we have $u_2=\TrLastChild(v)$.
\end{lemma}
\begin{proof}%
	In the case that $v$ is not a leaf in $T$, we claim that $u_2$ is the last child of $v$. Denote this child by $w$. Clearly $v$ is adjacent in $G$ to all of its children by definition. The parent of $w+1$ is larger than $v$, and thus $w+1$ cannot be adjacent to $v$ by the definition of $T$.
	
	If $v$ is a leaf of $T$, we claim that $u_2$ is the last child of the first internal (non-leaf) node before $v$ in level-order. Let $w=\TrLastChild(\TrPrevInt(v))$ denote this node. By definition, $s_w < v$ and $w \ge v$. As the neighborhood of $w$ forms a contiguous interval, $w$ is adjacent to $v$. Now consider $w+1$. By definition of $w$, its level-order successor $w+1$ must have parent $s_{w+1} > v$. Thus by the previous argument, it cannot be adjacent to $v$.
\end{proof}

\noindent
$\GDegree$: $|\GNeighbor(v)| = \GDegree(v)$ can be found in $O(1)$ time by computing $u_2-u_1$ for $u_1,u_2$ from $\GNeighbor(v)$.

\medskip\noindent
The results in this section are summarized in the following theorem;
we note that the succinct representation of neighbors allows to report those
faster than is possible using Acan et al's representation.

\begin{theorem}[Succinct proper interval graphs with distance]
	A connected proper interval graph can be represented in asymptotically optimal $2n + o(n)$ bits while supporting \GAdjacent, \GDegree, \GNeighbor and \GDistance in $O(1)$ time, and $\GSPath(u,v)$ in $O(\GDistance(u,v))$ time.
	A disconnected proper interval graph will use $3n + o(n)$ bits in the worst case; 
	if the number of components is $O(n/\log n)$, then the space is still $2n +o(n)$.
\end{theorem}

\subsection{Distances in Circular-Arc Graphs}

We finally show how to extend our distance oracles to circular-arc graphs.
We follow the notation of~\cite{AcanChakrabortyJoRao19} for circular-arc graphs,
in particular, we assume that we are given left and right endpoints of the vertices' arcs
in $[\ell_v,r_v] \in [2n]$ for $v=1,\ldots,n$, all endpoints are distinct, and 
$\ell_1<\cdots<\ell_n$, \ie, vertex ids are by sorted left endpoints.
Moreover, $v$ is a \emph{normal} vertex if $\ell_v < r_v$; otherwise it is a \emph{reversed} vertex
corresponding to the arc $[\ell_v,2n]\cup[1,r_v]$.
We assume that $G$ is connected; if not, $G$ is actually an interval graph, and we can use \wref{thm:succinct-interval-graphs}.

Acan et al.~\cite{AcanChakrabortyJoRao19,AcanChakrabortyJoRao19arxiv} 
describe two succinct data structures for circular-arc graphs: 
one based on succinct point grids (the ``grid version'') that supports all operations
of \wref{lem:succinct-interval}, but each with a $\Theta(\log n / \log \log n)$-factor
overhead in running time 
(see \cite[Thm.\,5]{AcanChakrabortyJoRao19} resp.\ 
\cite[Thm.\,6]{AcanChakrabortyJoRao19arxiv}),
and a second (the ``grid-less version'') that does not support \GDegree (other than by iterating over \GNeighbor),
but handles all other queries in optimal time 
(see \cite[Thm.\,7]{AcanChakrabortyJoRao19arxiv}).
We describe how to augment either of these to also answer \GDistance queries 
(in $\Oh(\log n / \log \log n)$ resp.\ $\Oh(1)$ time) 
using $\Oh(n)$ additional bits of space.

The idea of our distance oracle is to simulate access to the interval graph 
obtained by ``unrolling'' $G$ \emph{twice}, and then use the distance
algorithm for interval graphs therein.
\wref{fig:circular-arc-unrolled} shows an example.

\begin{figure}[tbh]
	\resizebox{\linewidth}!{\includegraphics{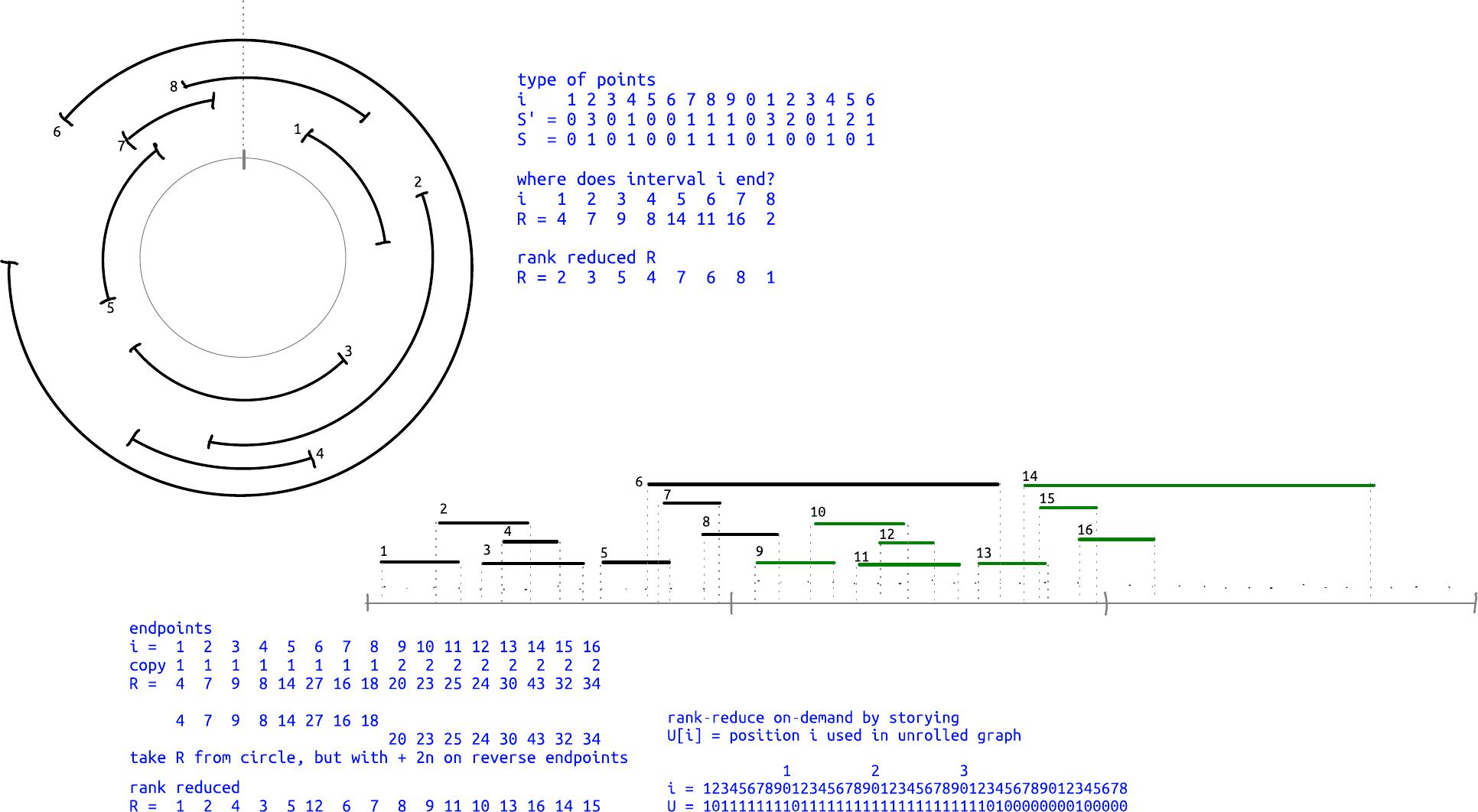}}
	\caption{%
		An examplary circular-arc graph and its twice-unrolled interval graph.
		The figure also shows some of the sequences used in Acan et al.'s succinct representations.
	}
	\label{fig:circular-arc-unrolled}
\end{figure}

Gavoille and Paul~\cite{GavoillePaul2008} have shown that this construction 
preserves distances in the following sense:
\begin{lemma}[{\cite[Lem.\,6]{GavoillePaul2008}}]
\label{lem:circular-arc-dist}
	Let $G = ([n],E)$ be a circular-arc graph with arcs $[\ell_v,r_v]$ 
	where endpoints are distinct and in $[2n]$ and $\ell_1<\cdots<\ell_n$.
	Define $\tilde G = ([2n],\tilde E)$ as the interval graph with the following sets
	of intervals: 
	for every normal vertex $v$, include $[\ell_v,r_v]$ and $[\ell_v+2n,r_v+2n]$ 
	and for every reversed vertex $u$, 
	include $[r_u,\ell_u+2n]$ and $[r_u+2n,\ell_u+4n]$.
	Then for any $u<v$, we have 
	(identifying vertices with the ranks of their left endpoints)
	\[
		\GDistance_G(u,v) \wwrel= 
		\min\bigl\{ \GDistance_{\tilde G}(u,v),\; \GDistance_{\tilde G}(v,u+n) \bigr\}.
	\]
\end{lemma}

Both data structures of Acan et al.\
store the sequences $r'$ and $r''$ of the 
rank-reduced right endpoints
for normal resp.\ reversed vertices, in the order of their left endpoints.
Using rank/select on the bitvectors $S$ and $S'$~-- storing the ``type''
of endpoints (left vs.\ right for $S$; 
left normal, right normal, left reversed, right reversed for $S'$)~--
we can compute the endpoints $(l_v,r_v)\in[2n]^2$ of any vertex $v$
in the same complexity as reading entries of $r'$ and $r''$,
\ie, $\Oh(\log n / \log \log n)$ time for the grid version and $\Oh(1)$ time for the grid-free version.

Given access to $r$, the sequence of right endpoints of the circular arcs,
we can simulate access to a right endpoint $\tilde r_v$, $v\in[2n]$, 
in the twice-unrolled interval graph $\tilde G$ as follows:
If $v\le n$ and a normal vertex, $\tilde r_v = r_v$.
If $v\le n$ and a reversed vertex, $\tilde r_v = r_v+2n$.
Otherwise, $v\in[n+1,2n]$; then $\tilde r_v = \tilde r_{v-n} + 2n$.
(See \texttt R in \wref{fig:circular-arc-unrolled}.)
By storing the bitvector $U[1..6n]$ with rank support where $U[i] = 1$ iff 
$\tilde{\ell_v} = i$ or $\tilde r_v = i$ for some $v$,
we can compute the rank-reduced intervals $[\tilde \ell'_v,\tilde r'_v]$
for all vertices $v=1,\ldots,2n$ of $\tilde G$.
We also store the distance tree for $\tilde G$ using the data structure of 
\wref{thm:levelorder-trees} in $4n+o(n)$ bits,
as well as the auxiliary data structures of Acan et al.\ (without $r$) from 
\wref{lem:succinct-interval}, all of which occupy $\Oh(n)$ bits.
Together this shows the following result.

\begin{theorem}
\label{thm:circular-arc}
	A circular-arc graph on $n$ vertices can be represented in $n \lg n + o(n\lg n)$ 
	bits of space to support either
	\begin{thmenumerate}{thm:circular-arc}
	\item 
		\GAdjacent, \GDegree, and \GDistance in $\Oh(\log n/ \log \log n)$ time,\\
		$\GNeighbor(v)$ in $\Oh((\GDegree(v)+1)\cdot \log n/ \log \log n)$, and\\
		$\GSPath(u,v)$ in $\Oh((\GDistance(u,v)+1)\cdot \log n/ \log \log n)$ time; or
	\item 
		\GAdjacent and \GDistance in $\Oh(1)$ time,\\
		$\GNeighbor(v)$ and $\GDegree(v)$ in $\Oh(\GDegree(v)+1)$, and\\
		$\GSPath(u,v)$ in $\Oh(\GDistance(u,v)+1)$ time.
	\end{thmenumerate}
\end{theorem}

\section{Conclusion}

We present succinct data structures and distance oracles for
interval graphs and several related families of graphs.
All are based on the solution of a fundamental data-structuring problem on trees: 
translating between breadth-first ranks and depth-first ranks of nodes in an ordinal tree.
Apart from demonstrating the unmatched versatility of tree covering~-- the only
method for space-efficient representations of trees known to support this BFS-DFS mapping~-- 
level-order operations are likely to find further applications 
in space-efficient data structures.

Regarding open questions, we note that one operation that is supported 
by standard tree covering has unwaveringly resisted all our attempts
to be realized on top of tree slabbing: generating $\lg n$ consecutive bits
of the \BP or \DFUDS of the tree.
Such operations are highly desirable as they allow immediate reuse
of any auxiliary data structures to support operations
on the basis of \BP resp.\ \DFUDS.
These sequences are inherently depth-first, though, and seem incompatible
with slicing the tree horizontally: the sought $\lg n$ bits might span a large number 
of (tier-2) slabs.
How and if level-order rank/select and generating a word of \BP or \DFUDS
can be simultaneously supported to run in constant time remains an open question.

\ifarxiv{
	\clearpage
	\appendix
	
	\section{Survey of Succinct Tree Representations}
\label{app:table}
A more complete survey of the previous representations of ordinal trees is given here, along with a table comparing the different techniques.

The \textsl{level-order unary degree sequence} ($\LOUDS$) representation of an ordinal tree~\cite{j1989}
consists of listing the degrees of nodes in unary encoding while traversing the tree
with a breadth-first search.
This is a direct generalization of the representation of heaps, \ie, 
complete binary trees stored in an array in breadth-first order:
There, due to the completeness of the tree, no extra information is needed to map
the rank of a node in the breadth-first traversal to the ranks of its 
parent and children in the tree.
The $\LOUDS$ is exactly the required information to do the same for general ordinal trees.
Historically one of the first schemes to succinctly represent a static tree,
$\LOUDS$ is still liked for its simplicity and practical efficiency~\cite{ArrovueloClanovasNavarro2010},
but a major disadvantage of $\LOUDS$-based data structures is that they support only 
a very limited set of operations~\cite{Navarro2016}.

Replacing the breadth-first traversal by a depth-first traversal yields the 
\textsl{depth-first unary degree sequence} (\DFUDS) encoding of a tree,
based on which succinct data structures with efficient support for many more operation have 
been designed~\cite{bdmrrr2005}.
Other approaches that allow to support largely the same set of operations are based on the 
\textsl{balanced-parentheses} (\BP) encoding~\cite{MunroRaman2001} or rely on \textsl{tree covering} (\TC)~\cite{GearyRamanRaman2006}
for a hierarchical tree decomposition.
%

As the oldest tree representation after \LOUDS, the \BP-based representations have a long history and the support for many operations was added for different applications. Munro and Raman~\cite{MunroRaman2001} first designed a $\BP$-based representation supporting $\TrParent$, $\TrNbDesc$, $\TrRank_{\pre/\post}$ and $\TrSelect_{\pre/\post}$ in $O(1)$ time and $\TrChild(x, i)$ in $O(i)$ time.
This is augmented by Munro~et al.~\cite{mrr2001} to support operations related to leaves in constant time, including $\TrLeafRank$, $\TrLeafSel$, $\TrLeftLeaf$ and $\TrRightLeaf$, which are used to represent suffix trees succinctly.
Later, 
Chiang~et al.~\cite{cll2005} showed how to support $\TrDeg$ using the $\BP$ representation in constant time which is needed for succinct graph representations,
while Munro and Rao~\cite{mr2004} designed $O(1)$-time support for $\TrLevAnc$, $\TrLevelPred$ and $\TrLevelSucc$ to represent functions succinctly.
Constant-time support for $\TrChild$, $\TrChildRank$, $\TrHeight$ and $\LCA$ is then provided by Lu and Yeh~\cite{ly2006}, that for $\TrRank_{\inorder}$ and $\TrSelect_{\inorder}$ by Sadakane~\cite{s2007} in their work of encoding suffix trees, and that for $\TrLevelLeft$ and $\TrLevelRight$ by Navarro and Sadakane~\cite{NavarroSadakane2014}.
        
Benoit~et al.~\cite{bdmrrr2005} were the first to represented a tree succinctly using $\DFUDS$, and their structure supports $\TrChild$, $\TrParent$, $\TrDeg$ and $\TrNbDesc$ in constant time. 
This representation is augmented by Jansson~et al.~\cite{JanssonSadakaneSung2012} to provide constant-time support for $\TrChildRank$, $\TrDepth$, $\TrLevAnc$, $\LCA$, $\TrLeafRank$, $\TrLeafSel$, $\TrLeftLeaf$ and $\TrRightLeaf$, $\TrRank_{\pre}$ and $\TrSelect_{\pre}$. 
To design succinct representations of labeled trees, Barbay~et al.~\cite{bhmr2011} further gave $O(1)$-time support for $\TrRank_{\DFUDS}$ and $\TrSelect_{\DFUDS}$. 

$\TC$ was first used by Geary~et al.~\cite{GearyRamanRaman2006} to represent a tree succinctly to support $\TrChild$, $\TrChildRank$, $\TrDepth$, $\TrLevAnc$, $\TrNbDesc$, $\TrDeg$, $\TrRank_{\pre/\post}$ and $\TrSelect_{\pre/\post}$ in constant time.
He~et al.~\cite{HeMunroRao2012} further showed how to use $\TC$ to support all other operations provided by $\BP$ and $\DFUDS$ representations in constant time, except $\TrRank_{\inorder}$ and $\TrSelect_{\inorder}$ which appeared after the conference version of their work.
Later, based on a different tree covering algorithm, Farzan and Munro~\cite{frr2009} destined a succinct representation that not only supports all these operations but also can compute an arbitrary word in a $\BP$ or $\DFUDS$ sequence in $O(1)$ time.
The latter implies that their approach can support all the operations supported by $\BP$ or $\DFUDS$ representations. 

	
	
	
	
	
	
	
	
	
	
	
	
%
	
	

\begin{table}[h]
\begin{center}
\begin{tabular} {|l|c|c|c|c|} \hline
operations                             &\BP 
&\DFUDS     &{\it previous} \TC
&our work  \\ \hline
$\TrChild$, $\TrChildRank$               
                         &\checkmark    &\checkmark  &\checkmark &\checkmark \\ \hline
$\TrDepth$, $\TrLevAnc$, $\LCA$  
                         &\checkmark    &\checkmark  &\checkmark &\checkmark \\ \hline
$\TrNbDesc$, $\TrDeg$
                         &\checkmark    &\checkmark  &\checkmark &\checkmark \\ \hline
$\TrHeight$              &\checkmark    &            &\checkmark &\checkmark \\ \hline
$\TrLeftLeaf$, $\TrRightLeaf$
                         &\checkmark    &\checkmark  &\checkmark &\checkmark \\ \hline
$\TrLeafRank$, $\TrLeafSel$
                         &\checkmark    &\checkmark  &\checkmark &\checkmark \\ \hline
$\TrLevelLeft$, $\TrLevelRight$
                         &\checkmark    &            &\checkmark &\checkmark \\ \hline
$\TrLevelPred$, $\TrLevelSucc$
                         &\checkmark    &            &\checkmark &\checkmark \\ \hline
$\TrRank_{\pre}$, $\TrSelect_{\pre}$   
                         &\checkmark    &\checkmark  &\checkmark &\checkmark \\ \hline
$\TrRank_{\post/\inorder}$, $\TrSelect_{\post/\inorder}$   
                         &\checkmark    &            &\checkmark &\checkmark \\ \hline
$\TrRank_{\DFUDS}$, $\TrSelect_{\DFUDS}$   
                         &              &\checkmark  &\checkmark &\checkmark \\ \hline
$\TrRank_{\lvl}$, $\TrSelect_{\lvl}$                         
                         &              &            &           &\checkmark \\ \hline
$\TrPrevInt$, $\TrNextInt$                         
                         &              &            &           &\checkmark \\ \hline
\end{tabular}
\end{center}
\caption{Operations supported in constant time by different succinct tree representations. }
\label{tab:suctree}
\end{table}

\section{Tree Operations}
\label{app:more-operations}
In this appendix, we sketch how to support the remaining operations from \wref{tab:operations}.
Many techniques are similar to previous work on \TC data structures~\cite{FarzanMunro2014,HeMunroRao2012,GearyRamanRaman2006},
but most operations require some changes to work on top of tree slabbing.
Operations required for our distance oracles are presented in full details
to be self-contained; for the others and where appropriate, 
we only describe the changes necessary to the algorithms
given in~\cite{FarzanMunro2014}.

\subparagraph{\TrParent:}
$\TrParent(v) = \TrLevAnc(v,1)$, so it is subsumed by the level-ancestor solution below.

\subparagraph{\TrLastChild:} 
Obviously, this can be obtained as $\TrLastChild(v) = \TrChild(v,\TrDeg(v))$ 
using the operations below, but it 
can also easily be implemented directly as follows.

	Given $\tau(v)$, find $\tau(u)$, for $u$ the rightmost child of $v$.
	Suppose first that $\tau_3(v) \ne 1$. Then all children of $v$ 
	are inside $\mu^{\tau_1(v)}_{\tau_2(v)}$.
	We use the micro-tree lookup table to obtain $\tau_3(u)$,
	and whether $u$ is a promoted node.
	If not, we return $\langle\tau_1(v),\tau_2(v),\tau_3(u)\rangle$.
	For promoted nodes, we store their canonical $\tau$-name.
	We store the canonical $\langle\tau_2(u),\tau_3(u)\rangle$
	if $u$ is in $\mu^{\tau_1(v)}$, and the full $\tau(u)$ otherwise,
	plus 1 extra bit to distinguish these cases.
	(This amounts to $o(n)$ extra bits as there are 
	$\Oh(n/H')$ tier-2 promoted nodes
	and $\Oh(n/H)$ tier-1 promoted nodes.)
	
	If $\tau_1(v)=1\ne\tau_2(v)$, we store $\langle\tau_2(u),\tau_3(u)\rangle$
	of $v$' rightmost child, which must lie in $\mu^{\tau_1(v)}$.
	If $\tau_1(v)=1=\tau_2(v)$, we simply store $\tau(u)$ directly.

\subparagraph{\TrDepth:}
Given $\tau(v)$, compute the level on which $v$ lies.
We store the global depth of the mini-tree root and the mini-tree-local
depth at each micro-tree root. For a node $v$, find the depth relative to the micro-tree
root using the lookup table, and add the mini-tree-local depth and the global depth. We
may need to adjust for dummy roots but that is trivial.

\subparagraph{\TrLevAnc:}
Given $\tau(v)$, find $\TrLevAnc(v,i)= \tau(w)$ for $w$ the ancestor of $v$
on level $\TrDepth(v)-i$.
The solution of~\cite[\S3]{GearyRamanRaman2006} essentially works without changes, 
but tree slabbing actually simplifies it slightly.
We start by bootstrapping from a non-succinct solution 
for the level-ancestor (LA) problem:
\begin{lemma}[Level ancestors, {\cite[Thm.\,13]{BenderFarachColton2004}}]
	There is a data structure using $\Oh(n \log n)$ bits of space 
	that answers $\TrLevAnc(v,i)$ queries on a tree of $n$ nodes 
	in $\Oh(1)$ time.
\end{lemma}
Geary et al.\ apply this to a so-called macro tree; we observe that we can instead
build the LA data structure for all tier-1 s-nodes, 
where s-nodes $u$ and $v$ are connected by a macro edge 
if there is a path from $u$ to $v$ in $T$ that does not contain further s-nodes.
This uses $\Oh(\frac nH \log(\frac nH)) = \Oh(n/\log n)$ bits.
Each mini-tree root stores its closest ancestor that is a tier-1 s-node.
Additionally, mini/micro tree roots and (tier-1/tier-2) s-nodes 
store collections of \emph{jump pointers}:
mini trees / tier-1 s-nodes allow to jump to an ancestor at any distance in 
$1,2,\ldots,\sqrt{H}$ or $\splitaftercomma{\sqrt{H},2\sqrt{H},3\sqrt{H},\ldots,H}$;
the same holds for micro trees / tier-2 s-nodes with $H'$ instead of $H$, and
as usual storing only $\langle \tau_2,\tau_3\rangle$.
(Mini-tree roots / tier-1 s-nodes store full $\tau$-names in jump pointers.)

The query now works as follows (essentially \cite[Fig.\,6]{GearyRamanRaman2006}, 
but with care for s-nodes):
We compute the micro-tree local depth of $v$ by table lookup 
and check if $w$ lies inside the micro tree; if so, we find it by table lookup.
If not, we move to the micro-tree root~-- 
or the tier-2 s-node in case the micro-tree root is a dummy root 
(using a micro-tree local \TrLevAnc query); let's call this node $x$.
We now compute $x$'s mini-tree local depth
(using the data structures for \TrDepth)
to check if $w$ lies inside this mini-tree.
If it does,
we use $x$'s jump pointers: either directly to $w$
(if the distance was at most $\sqrt{H'}$), or to get within distance $\sqrt{H'}$, from where we continue recursively.
If $w$ is not within the current mini-tree, we jump to $y$, 
the mini-tree root, or a tier-1 s-node in case the mini tree has a dummy root
(using a recursive, mini-tree local \TrLevAnc query).
If $w$ is within distance $H$ from there,
we use $y$'s jump pointers 
(to either get to $y$ directly, or to get within distance $\sqrt{H}$).
Otherwise, we use $y$'s pointer to its next tier-1 s-node ancestor 
(unless $y$ already is such).
The LA data structure on tier-1 s-nodes allows us to jump within distance $H$ of~$w$,
from where we continue.

Note that after following two root jump pointers of each kind 
we are always close enough to $w$
that the next micro-tree root will have a direct jump pointer to $w$.
The recursive call to find a tier-1 s-node subforest root 
(when a mini-tree has a dummy root) is always resolved local to the mini tree,
so cannot lead to another such recursive calls.
Hence the running time is~$\Oh(1)$.

\separatedpar
The remaining tree operations are not immediately needed for the computation of 
distances in interval graphs.
We sketch how to support the operations by describing the changes needed to make to the approach used in previous work of \TC.

\subparagraph{\textbf{\TrChild, \TrChildRank}:} 
For \TrChild, no changes are necessary, as we will never be getting a child of a dummy root. As for \TrChildRank, the only difference occurs when we need to find the rank of an $s-node$. Its rank in the mini(micro)-tree is wrong because of the dummy root. For the tier-1 $s$-nodes, we store a bit-vector storing a 1 whenever the preceding $s$-node has a different parent. The \TrChildRank would be distance to the preceding 1 in the bit-vector. The length of the bit-vector is the number of tier-1 $s$-nodes which is $O(n/H)$. Similarly for tier-2 $s$-nodes.

\subparagraph{\textbf{\TrDeg, \TrNbDesc}:} 
No changes are necessary for \TrDeg or \TrNbDesc.

\subparagraph{\textbf{\TrHeight}:} 
For a mini-tree root, we may explicitly store the height. For each tier-1 $s$-node, we may also explicitly store the height. Now we describe how to find the height of a micro-tree root. For a micro-tree root, we store the micro-tree that contains the deepest descendant. If this micro-tree has a tier-1 promoted $s$-node, we store the promoted $s$-node with the greatest height. The height of the micro-tree root can be found by the difference in depths of the two micro-tree roots, plus the height of the tier-1 $s$-node. For a node that is not a micro-tree root, we consider the micro-tree $\mu^i_j$ that it is in. Suppose that $\mu^i_j$ does not contain any tier-2 promoted $s$-nodes. Then we proceed in the same way as in \cite{FarzanMunro2014}. Otherwise, using the lookup table, we find the range of tier-2 promoted $s$-nodes that are descendants, and using a range-maximum query, find the tier-2 promoted $s$-node that has the greatest depth. To find the depth of a tier-2 $s$-node, we store the micro-tree containing the deepest descendant as in the root case. We then proceed in the same manner. The space required for range-maximum queries on all tier-2 $s$-nodes is linear in the number which is $O(n/H')$.

\subparagraph{\textbf{\TrLeftLeaf, \TrRightLeaf}:} 
This is done in the same way as previously. The only difference is that we need to store the left most/right most leaf at every tier-1 $s$-node. We also need to store the micro-tree that contains the left most/right most leaf, or the micro-tree containing the relevant tier-1 $s$-node at every tier-2 $s$-node.

\subparagraph{\textbf{\TrLeafSize}:} 
At each tier-1 $s$-node we store the number of leaves in the subtree rooted at the $s$-node. We also store the prefix sum of these values (the sum of the number of leaves from the first $s$-node to the current $s$-node). For tier-2 $s$-nodes, we store the number of leaves in the subtree of the mini-tree rooted at the $s$-node. We do not include tier-1 $s$-nodes (which are leaves of the mini-tree) in this count. For the $s$-nodes of each mini-tree, we store the prefix-sum of the number of leaves (starting from the first tier-2 $s$-node of the mini-tree to the current $s$-node).

To find the number of leaves below a node, we find the number of leaves in the micro-tree using the lookup table. We find the range of the tier-2 $s$-nodes below it, if the micro has any tier-2 promoted $s$-nodes. If not we check the unique outgoing edge if necessary for tier-2 promoted $s$-nodes. From the range of the promoted $s$-nodes, we sum of the leaves in the mini-tree from the prefix sum data structure. We also find the tier-1 $s$ nodes below in similar fashion. We then take the sum of the sizes of the tier-1 $s$-nodes using the prefix-sum data structure.
\textbf{\TrLeafRank and \TrLeafSel}: \TrLeafSel is done in the same way as before, using the compressed bit vector approach. For \TrLeafRank, in addition to the information stored, we also need to store the number of leaves preceding tier-1 $s$-nodes. For tier-2 $s$-nodes, we store the preceding tier-1 $s$-node, and the number of leaves between them.

\subparagraph{\textbf{\TrLevelLeft, \TrLevelRight}:} 
No changes needed \wrt previous work.

\subparagraph{\textbf{\TrLevelSucc, \TrLevelPred}:} 
Using $\TrRank_{\lvl}$ and $\TrSelect_{\lvl}$, these operations are now straight-forward
and do not need a tailored implementation.

\subparagraph{\textbf{\LCA}:} 
The technique of He et al.~\cite{HeMunroRao2012} works for tree slabbing, too. The only change we need to make is to include tier-1 $s$-nodes in the tier-1 macro tree and tier-2 $s$-nodes in each tier-2 macro tree. These will be included instead of the dummy root added.

\subparagraph{\textbf{$\TrRank_{\pre/\post/\inorder/\DFUDS/\lvl}$, $\TrSelect_{\pre/\post/\inorder\DFUDS/\lvl}$}:} 
For other traversals can be handled similarly to preorder / level order.

}{}

\bibliography{references}

\end{document}